\documentclass[a4paper,UKenglish,cleveref, autoref, thm-restate]{lipics-v2021}

\usepackage{tikz}
\usetikzlibrary{automata}

\usepackage{mdframed}
\newmdtheoremenv{theo}{Theorem}
\newmdtheoremenv{idef}{Definition}

\newcommand{\N}{\mathbb{N}}

\title{Positional Determinacy with Colored Vertices: a $1$-to-$2$-Player Lift} 

\author{Rapha\"el {Berthon}}{Université Paris-Saclay, CNRS, ENS Paris-Saclay, Laboratoire Méthodes Formelles, 91190, Gif-sur-Yvette, France \and \url{https://home.lmf.cnrs.fr/Rberthon/} }{rberthon@lmf.cnrs.fr}{https://orcid.org/0000-0002-2580-5193}{[Funded by ANR-22-CE48-0012 (Bisous)]}

\author{St\'ephane {Le Roux}}{Université Paris-Saclay, CNRS, ENS Paris-Saclay, Laboratoire Méthodes Formelles, 91190, Gif-sur-Yvette, France \and  \url{https://home.lmf.cnrs.fr/StephaneLeRoux/}}{stephane.le_roux@ens-paris-saclay.fr}{https://orcid.org/0000-0002-6511-0572}{[Funded by ANR-22-CE48-0012 (Bisous)]}

\authorrunning{R. Berthon and S. Le Roux} 

\Copyright{Rapha\"el Berthon and St\'ephane Le Roux} 



\category{} 

\relatedversion{} 

\acknowledgements{Thanks to Bastien Laboureix for useful conversations.}

\nolinenumbers 

\EventEditors{Ana Sokolova and Patrick Totzke}
\EventNoEds{2}
\EventLongTitle{37th International Conference on Concurrency Theory (CONCUR 2026)}
\EventShortTitle{CONCUR 2026}
\EventAcronym{CONCUR}
\EventYear{2026}
\EventDate{September 1--4, 2026}
\EventLocation{Liverpool, UK}
\EventLogo{}
\SeriesVolume{391}
\ArticleNo{34}

\begin{document}
\maketitle

\begin{abstract}
Positional determinacy of vertex-colored parity games was proved in the 1990s, which directly implies positional determinacy of edge-colored parity games. In 2006, it was shown that if a prefix-independent color-based objective ensures that every edge-colored two-player turn-based game is positionally determined, this objective is equivalent to a parity objective. We prove a similar result for vertex-colored games, namely that the following are equivalent for any prefix-independent objective $W$ over a finite set of colors:
\begin{itemize}
    \item $W$ is positionally determined on all vertex-colored one-player games. 
     
    \item $W$ is positionally determined on all vertex-colored two-player games. 
    
    \item $W$ is equivalent to a parity objective on ordrerd pairs of colors.
\end{itemize}
We prove that finiteness of the color set is required for our equivalence to hold. Beyond this $1$-to-$2$-player lift, the technique that we develop to handle the pairs of colors establishes a promising 2-way correspondence between edge-colored games and vertex-colored games. 

\end{abstract}

\section{Introduction}

\textbf{Context: }Game theory is applied to many fields such as economics, political science,
evolutionary biology, and is used for model-checking processes in the industry~\cite{DBLP:conf/cav/Cook18}. While many game-theoretic models allow for multiple players and non-zero-sum payoffs, fundamental games in logic and computer science usually involve two players in a win-lose setting, where every play results in a win for exactly one player.

In model checking, these win-lose games are typically turn-based and played on a finite or infinite labeled graph, called an arena. Starting from an initial vertex, the player owning the current vertex selects an outgoing edge to reach the next vertex. This continues indefinitely, producing an infinite path.
A player's strategy maps the history of visited vertices to her next move. Specifically, for every finite path ending at a vertex owned by that player, the strategy selects the edge to be followed. A pair of one strategy per player induces a unique infinite sequence of edges and vertices starting from the initial vertex, constituting the play.

General strategies can be complex to implement: in a play of infinite duration, the number of distinct histories is infinite; therefore, implementing such a strategy generally requires a mechanism to store and process a history of unbounded length. Much work has thus focused on strategies that require only finite memory, usually represented as a finite-state machine, or no memory at all. In the latter case, strategies map each position in the arena to a single outgoing edge, and are said to be positional (or memoryless). 

The objective of the game is defined independently of the arena, as a set of infinite words $W\subseteq \Gamma^\omega$, where $\Gamma$ is a set of colors. Either vertices or edges are assigned colors from $\Gamma$, to obtain vertex-colored arenas (or games) or edge-colored arenas (or games). In both cases, every play generates an infinite sequence of colors called a trace. The player called Eve wins if this trace is in $W$, and Adam wins otherwise.

To state our results, we use the concept of \emph{parity objective}: an infinite sequence of integers, chosen from a finite set called the set of priorities, satisfies the parity objective if the greatest integer occurring infinitely often is even. In~\cite{CN06}, the authors considered \emph{generalized parity objectives}: intuitively, they are defined on infinite sequences of colors but are similar to the parity objective via a \emph{priority function} from colors to integers.

In the context of infinite-duration games, prefix-independence is an important property for win-lose objectives as well as for more general objective functions (like gain functions). An objective is prefix-independent, also called tail, if the winner of a play is unaffected by the addition or removal of any finite initial sequence of colors. This property is central to the analysis of infinite games: e.g. if there is a winning strategy from a given vertex, then there exists a winning strategy from every vertex from which the player can reach this given vertex.

An objective is said to be positionally determined on a class of arenas if a player with a winning strategy from a given vertex also has a positional one from that vertex. Colcombet and Niwi{\'n}ski showed in~\cite{CN06} that prefix-independent objectives that are positionally determined on both finite and infinite edge-colored arenas are exactly those equivalent to the parity objective. The case for vertex-colored games has been opened until now.

\begin{figure}[ht]
    \centering
    \begin{subfigure}{.45\textwidth}
        \centering
        \begin{tikzpicture}[shorten >=1pt, node distance=2cm, auto, bend angle=45]
            \node[state] (v) {};
            \path[->] (v) edge [loop left] node {$a$} (v)
                      (v) edge [loop right] node {$b$} (v);
        \end{tikzpicture}
        \subcaption{Edge-colored hub-cycle arena}\label{subfig:edge-col}
    \end{subfigure}
    \hfill
    \begin{subfigure}{.45\textwidth}
        \centering
        \begin{tikzpicture}[shorten >=1pt, node distance=2.5cm, auto]
            \node[state] (a) {$a$};
            \node[state] (b) [right of=a] {$b$};
            \path[->] (a) edge [loop left] node {} (a)
                      (a) edge [bend left] node {} (b)
                      (b) edge [bend left] node {} (a);
        \end{tikzpicture}
        \subcaption{Vertex-colored hub-cycle arena}\label{subfig:vertex-col}
    \end{subfigure}
    \caption{Two arenas where Eve controls all vertices.}
    \label{fig:hub-cycle}
\end{figure}
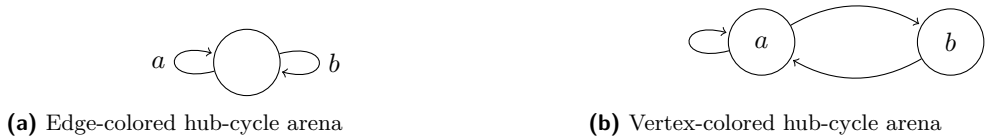

\begin{example}\label{ex:intro-ab}
Consider the objective $W = (a+b)^*(ab)^\omega$ on $\Gamma=\{a,b\}$, defining words that eventually alternate between $a$ and $b$. It will be our running example throughout the paper.

\noindent \emph{Edge-colored case.} In the edge-colored game defined by $W$ and the arena from Figure~\ref{subfig:edge-col} (which we will call a hub-cycle game), Eve, who controls the unique vertex, has a winning strategy, but no positional winning strategy: a non-positional strategy can alternate between edges $a$ and $b$, producing the trace $(ab)^\omega$, but any positional strategy for Eve must choose a single edge $e \in \{a, b\}$ to play infinitely; this eventually results in either $a^\omega$ or $b^\omega$. 

\noindent \emph{Vertex-colored case.} By contrast, $W$ is positionally determined on the vertex-colored arena from Figure~\ref{subfig:vertex-col}. Eve can play the positional strategy visiting $a$ from $b$ and visiting $b$ from $a$. Starting from $a$, this produces the infinite sequence $(ab)^\omega \in W$. We will later see that $W$ is positionally determined on all vertex-colored arenas.
\end{example}

\textbf{Contribution: }
 In this article we distinguish two settings: generalized parity objectives defined on sequences of colors (as in~\cite{CN06}), and also those defined on sequences of ordered pairs of colors.

Furthermore, we identify a specific class of strongly connected arenas already occurring in~\cite{CN06}, which we call \emph{hub-cycle arenas}, that contain at most one vertex with more than one outgoing edge. In this article, colors are assigned either to the edges or to the vertices. A \emph{hub-cycle game} is a one-player game played on such an arena. Figure~\ref{fig:hub-cycle} gives two examples of hub-cycle arenas.

Let $\Gamma$ be a set of colors. Our main contributions characterize positional determinacy via these hub-cycle arenas:

\begin{enumerate}
    \item If $\Gamma$ is finite, the following are equivalent for any prefix-independent objective $W\subseteq\Gamma^\omega$:
    \begin{enumerate}
        \item $W$ is positionally determined on all vertex-colored hub-cycle (one-player) arenas.\label{intro:2a}
        \item $W$ is positionally determined on all vertex-colored two-player arenas.\label{intro:2b}
        \item $W$ is a generalized parity objective on pairs.\label{intro:2c}
    \end{enumerate}
\end{enumerate}

While $1$-to-$2$-player lifts for edge-colored arenas already appear in the literature~\cite{GZ05,BLORV22,BRV23,CO26}, this is to our knowledge the first such result for vertex-colored arenas. We also prove that we cannot simply drop the assumption that $\Gamma$ is finite. Several works consider half-positionality~\cite{Kopczynski06, DBLP:journals/theoretics/Ohlmann23, CO26}, where Eve has a positional winning strategy whenever she has a winning strategy. We discuss related works in more details in Section~\ref{sec:rel}.

\textbf{Outline.} Section~\ref{sec:def} provides preliminary definitions. 
Section~\ref{sec:edge} recalls the existing $1$-to-$2$-player lift for edge-colored arenas. 
Vertex-colored arenas are addressed in Section~\ref{sec:vertex-words} and Section~\ref{sec:vertex-pairs} via reductions to generalized parity objectives on anchored words and then pairs. 
Section~\ref{sec:remarks} compares our results with existing literature, while Section~\ref{sec:rel} proposes future directions. Finally, Section~\ref{sec:conclusion} offers concluding remarks.

\section{Definitions}
\label{sec:def}

Since our results extends the results of \cite{CN06}, we broadly adopt their notational framework to facilitate the comparison with the original proofs and definitions. 

Let $\mathbb{N}=\{0,1,2\dots\}$ be the set of natural numbers. For $i,j\in \mathbb{N}$, $i\leq j$ we denote by $[i,j]\subseteq \mathbb{N}$ the set $\{i,i+1,\dots, j\}$. Let $\Gamma$ be a non-empty, possibly infinite set. Let $\Gamma^*$, $\Gamma^+$, and $\Gamma^\omega$ respectively denote the sets of finite, non-empty finite, and infinite words over $\Gamma$, with $\varepsilon$ representing the empty word. For $a \in \Gamma$, let $a^+ := \{a\}^+$. For $E \subseteq \Gamma^*$ and $W \subseteq \Gamma^* \cup \Gamma^\omega$, let $EW$ denote the concatenation $\{uv \mid u \in E, v \in W\}$. The complement of $W \subseteq \Gamma^\omega$ is $\overline{W} := \Gamma^\omega \setminus W$. For $u \in \Gamma^+$, $u^\omega := uuu\dots \in \Gamma^\omega$, and for $E \subseteq \Gamma^+$, we take $E^\omega = \{u_0u_1u_2\dots\ |\ \forall n \in \N, u_n \in E\}$.

\begin{definition}[Arenas]\label{def:arena}
Let $\Gamma$ be a non-empty, possibly infinite set of colors. A \textbf{two-player edge-colored arena} is a tuple $A = \langle V_E, V_A, E \rangle$, where the set of vertices $V = V_E \cup V_A$ is the disjoint union of the sets of vertices $V_E$ and $V_A$, belonging to Eve and Adam, respectively. $E \subseteq V \times \Gamma \times V$ is the set of colored edges. We assume that every vertex $v \in V$ has at least one outgoing edge. The class of all such arenas is $\mathcal{A}_e$. An arena with either $V_E=\varnothing$ or $V_A=\varnothing$ is a \textbf{one-player} arena. 

The class $\mathcal{A}_v \subsetneq \mathcal{A}_e$ of \textbf{vertex-colored arenas} contains all arenas satisfying the following constraint: for all $v, v', v'' \in V$ and $a, b \in \Gamma$, if $(v, a, v') \in E$ and $(v, b, v'') \in E$, then $a = b$. This amounts to a coloring on vertices. 
\end{definition}

\begin{definition}[Games and Plays]\label{def:game}
A \textbf{game} is a tuple $G = \langle A, W \rangle$ where $A \in \mathcal{A}_e$ is an arena and $W \subseteq \Gamma^\omega$ is a winning objective. 
A \textbf{play} in $G$ is an infinite sequence of edges $\pi = (v_0, a_0, v_1, a_1, \dots)$ such that $(v_i, a_i, v_{i+1}) \in E$ for all $i \in \mathbb{N}$. The resulting trace is the infinite word $w_\pi = a_0 a_1 \dots \in \Gamma^\omega$. Eve wins the play if $w_\pi \in W$; otherwise, Adam wins. We distinguish between \textbf{edge-colored} and \textbf{vertex-colored games} based on the class of arenas upon which they are defined, and also define \textbf{one-player games} that are played on a one-player arena.
\end{definition}

\begin{definition}[Strategies]
Let $A = \langle V_E, V_A, E \rangle$ be an arena. A \textbf{history} is a finite path $h = (v_0, a_0, v_1, a_1, \dots, v_n)$ such that $(v_i, a_i, v_{i+1}) \in E$ for all $0 \leq i < n$. The set of all histories is $\mathcal{H}$, and $\mathcal{H}_X \subseteq \mathcal{H}$ denotes histories ending in $v_n \in V_X$ for player $X \in \{E, A\}$.

A \textbf{strategy} for player $X$ is a function $\sigma_X: \mathcal{H}_X \to E$ such that for any $h = (\dots, v_n)$, $\sigma_X(h)$ is an edge $(v_n, a, v) \in E$. A play $\pi$ is \textbf{consistent} with strategy $\sigma_X$ if for every $i \in \mathbb{N}$ where $v_i \in V_X$, the transition satisfies $(v_i, a_i, v_{i+1}) = \sigma_X(v_0, a_0, \dots, v_i)$.

A strategy $\sigma_X$ is \textbf{positional} if it depends only on the current vertex. It is represented as a function $\sigma_X: V_X \to E$ where $\sigma_X(v)$ is an outgoing edge from $v$. 
A strategy $\sigma_X$ is \textbf{winning for player $X$} from an initial vertex $v_0$ if every play $\pi$ starting at $v_0$ and consistent with $\sigma_X$ is won by $X$. In this case, $v_0$ is a \textbf{winning vertex} for player $X$.
\end{definition}

\begin{definition}[Positional determinacy]
\textbf{A game} $G = \langle A, W \rangle$ is positionally determined if, from any initial vertex $v_0 \in V$, one of the players has a positional winning strategy. 

\textbf{An objective} $W \subseteq \Gamma^\omega$ is positionally determined on a class of arenas $\mathcal{A} \subseteq \mathcal{A}_e$ if for every arena $A \in \mathcal{A}$, the game $\langle A, W \rangle$ is positionally determined.
\end{definition}

\begin{definition}[Prefix-independence]
An objective $W \subseteq \Gamma^\omega$ is \textbf{prefix-independent}, also called \textbf{uniform} or a \textbf{tail-objective} if $\Gamma W = W$.
\end{definition}

\begin{remark}
For prefix-independent objectives, if a game is positionally determined for a player, then this player has a single positional strategy that wins from every vertex from where she can win. See Lemma 2.12 in~\cite{DBLP:conf/dagstuhl/2001automata} or Lemma 5 in~\cite{CN06} for a complete proof.  
\end{remark}

\begin{definition}[Parity Objectives]\label{def:parity}
Let $n \in \mathbb{N}$. 
\begin{itemize}
    \item $W \subseteq [0,n]^\omega$ is a \textbf{parity objective} (of order $n$) if a word $u = u_0 u_1 \dots \in W$ iff $\limsup_{i \to \infty} u_i$ is even. We refer to $[0,n]$ as the set of priorities. 
    \item $W \subseteq \Gamma^\omega$ is a \textbf{generalized parity objective} (of order $n$) if there exists a priority function $p: \Gamma \to \{0, \dots, n\}$ such that a word $u = u_0 u_1 \dots \in W$ iff $\limsup_{i \to \infty} p(u_i)$ is even.
    \item $W \subseteq \Gamma^\omega$ is a \textbf{generalized parity objective on pairs} (of order $n$) if there exists a priority function $p: \Gamma^2 \to \{0, \dots, n\}$ such that $u \in W$ iff $\limsup_{i \to \infty} p(u_i, u_{i+1})$ is even.
\end{itemize}
\end{definition}

\begin{example}\label{ex:def-ab}
We again consider the objective $W = (a+b)^*(ab)^\omega$  on $\Gamma=\{a,b\}$. This is a \textbf{generalized parity objective on pairs} of order $1$. Indeed, we can define the priority function $p: \{a, b\}^2 \to \{0, 1\}$ such that $p(a, b) = p(b, a) = 0$ and $p(a, a) = p(b, b) = 1$.
\end{example}

\section{Existing results on edge-colored games}
\label{sec:edge}
For $\Gamma$ a non-empty possibly infinite set of colors, and $W \subseteq \Gamma^\omega$, we define the set of winning finite cyclic words as $W_f := \{u \in \Gamma^+ \mid u^\omega \in W\}$. 

 We consider $\Gamma$, a possibly infinite set of colors. We recall the proof of the following, which includes a result demonstrated in Kopczy\'{n}ski's thesis~\cite{Kopczynski08}, that extends~\cite{CN06} to one player games: 

\begin{enumerate}\setcounter{enumi}{1}
    \item For any prefix-independent objective $W\subseteq \Gamma^\omega$, the following are equivalent:
    \begin{enumerate}
        \item $W$ is positionally determined on all edge-colored one-player arenas.\label{sec1:1a}
        \item $W$ is positionally determined on all edge-colored two-player arenas.\label{sec1:1b}
        \item $W$ is a generalized parity objective.\label{sec1:1c}
    \end{enumerate}
\end{enumerate}

\textit{Sketch of proof: } The implication \ref{sec1:1b} $ \Rightarrow $ \ref{sec1:1a} follows immediately as one-player arenas are a subclass of two-player arenas. The implication \ref{sec1:1c} $ \Rightarrow $ \ref{sec1:1b}, was established in the 1990s \cite{EJ91,Mostowski91,Zielonka98}, showing that generalized parity objectives are positionally determined on vertex-colored arenas; this result extends directly to edge-colored arenas. In 2006, \cite{CN06} proved the converse, i.e. \ref{sec1:1b} $ \Rightarrow $ \ref{sec1:1c}. Finally, \ref{sec1:1a} $ \Rightarrow $ \ref{sec1:1c} has been shown in 2008 by Kopczy\'{n}ski~\cite{Kopczynski08}. 

\begin{example}
We come back to our running example, the objective $W = (a+b)^*(ab)^\omega$ on $\Gamma=\{a,b\}$. Example~\ref{ex:intro-ab} shows that $W$ is not positional on the edge-colored arena defined in Figure~\ref{subfig:edge-col}. By the equivalence between \ref{sec1:1b} and \ref{sec1:1c} above, $W$ is not a generalized parity objective. 
\end{example}

The proof of our Lemma~\ref{lem:1p-pos} is very similar to the one in~\cite{Kopczynski08}, in particular to the following proposition:

\begin{restatable}{proposition}{onePosEdge}\label{prop:1p-pos-edge}
For any set of colors $\Gamma$, suppose $W\subseteq \Gamma^\omega$ is prefix-independent and positionally determined on edge-colored one-player arenas. 
Then, for any $L, L' \subseteq \Gamma^*$:
\begin{align*}
&\forall u \in L, \exists v \in L', uv \in W_f \iff \exists v \in L', \forall u \in L, uv \in W_f \\
&\forall u \in L, \exists v \in L', uv \in \overline{W}_f \iff \exists v \in L', \forall u \in L, uv \in \overline{W}_f
\end{align*}
\end{restatable}

\section{Vertex-colored games: parity on words}
\label{sec:vertex-words}
Throughout the paper, we focus on a restricted class of one-player arenas that serves as a building block for our results.

\begin{definition}[Hub-cycle arenas]\label{def:hub-cycle}
A \textbf{hub-cycle arena} is a strongly connected, one-player arena where at most one vertex has an out-degree strictly greater than one. A game is \textbf{hub-cycle} if its underlying arena is a hub-cycle arena, and can be edge-colored or vertex-colored.
\end{definition}

We find that positional determinacy for edge-colored two-player games is equivalent to its restriction to edge-colored hub-cycle games (which are one-player games), as stated by our characterization.

To study vertex-colored games, we begin by introducing a generalization of parity objectives, where priorities are assigned to finite words on colors instead of individual colors.
For $\Gamma$ a non-empty possibly infinite set of colors, and $W \subseteq \Gamma^\omega$, and for every letter $a \in \Gamma$, we define the set of winning finite cyclic words anchored at $a$ as $W_{a,f} = \{u \in a\Gamma^* \mid u^\omega \in W\}\subseteq W_f$.

\begin{definition}[Generalized parity objective on anchored words]\label{def:parity-words}
Let $\Gamma$ be a non-empty set of colors and $W \subseteq \Gamma^\omega$. We say $W$ is a \textbf{generalized parity objective on anchored words} if for every letter $a \in \Gamma$, there exists $n \in \mathbb{N}$ and a priority function $p_a: a\Gamma^* \to \{0, \dots, n\}$ such that the following holds:

For every word $w \in \Gamma^\omega$ where $a$ occurs infinitely often, let $a w_0, a w_1, a w_2, \dots$ (with $w_i \in \Gamma^*$) be any sequence of factors of $w$ such that $w = x a w_0 a w_1 a w_2 \dots$ for some $x \in \Gamma^*$.
Then $w\in W$ if and only if $\limsup_{i \to \infty} p_a(a w_i)$ is even.
\end{definition}

Note that the priority functions $(p_a)_{a \in \Gamma}$ of a generalized parity objective on anchored words do not specify membership in $W$ for sequences where no colors occur infinitely often. If $\Gamma$ is finite, they fully specify $W$, though. Also note that in general, these functions must be consistent with each other, e.g. when two colors occur infinitely often in a sequence.

\begin{example}
We illustrate generalized parity objectives on anchored words with our running example $W = (a+b)^*(ab)^\omega$ on $\Gamma=\{a,b\}$. Define $p_a(au) = 0$ if $u \in b(ab)^*$ and $p_a(au) = 1$ otherwise; define $p_b(bu) = 0$ if $u \in a(ba)^*$ and $p_b(bu) = 1$ otherwise. If a word $w$ contains infinitely many $a$, we present $w$ as $w = x a w_0 a w_1 a w_2 \dots$, and otherwise as $w = x b w_0 b w_1 b w_2 \dots$. In both cases, the even priority $0$ is only achieved by traces that eventually only alternate between $a$ and $b$.

\end{example}

We add the condition~\ref{sec2:2c} to our equivalence on vertex-colored games below. Condition~\ref{sec2:2c} is a useful and informative intermediate step from condition~\ref{sec2:2a} to condition~\ref{sec2:2d}.

\begin{enumerate}
\item If $\Gamma$ is finite, the following are equivalent for any prefix-independent objective $W\subseteq \Gamma^\omega$:
\begin{enumerate}
    \item $W$ is positionally determined on all vertex-colored hub-cycle one-player arenas.\label{sec2:2a}
    \item $W$ is positionally determined on all vertex-colored two-player arenas.\label{sec2:2b}
    \item $W$ is a generalized parity objective on pairs.\label{sec2:2d}
    \item $W$ is a generalized parity objective on anchored words.\label{sec2:2c}
\end{enumerate}
\end{enumerate}

As in the edge-colored case, the implication \ref{sec2:2b} $ \Rightarrow$ \ref{sec2:2a} is direct since hub-cycle games are a subclass of two-player games. 
We prove \ref{sec2:2a} $ \Rightarrow $ \ref{sec2:2c} by progressing up to Lemma~\ref{post-det-full-conseq} using techniques similar to the edge-colored proof in~\cite{CN06}.
Our proof holds even when $\Gamma$ is infinite: 

\begin{restatable}{theorem}{revcase}\label{thm:pose-parity-words}
    \begin{mdframed}[innertopmargin=+0.5em,innerbottommargin=+0.5em]
    Let $\Gamma$ be a non-empty possibly infinite set of colors. If $W \subseteq \Gamma^\omega$ is prefix-independent and positionally determined on vertex-colored hub-cycle arenas, then $W$ is a generalized parity objective on anchored words.
    \end{mdframed}
\end{restatable}

We start with general results on languages, then show properties implied by prefix-independence and positional determinacy.

\begin{lemma}\label{lem:na-Wf} Let $W \subseteq \Gamma^\omega$.
	\begin{enumerate}
		\item\label{lem:na-Wf1} $(\overline{W})_f = \Gamma^+ \setminus W_f$.
		\item\label{lem:na-Wf2} Assume $\forall u,v \in \Gamma^+, uv \in W_f \Rightarrow vu \in W_f$. Then $\forall u,v \in \Gamma^+,uv \in (\overline{W})_f \Leftrightarrow vu \in (\overline{W})_f$.
	\end{enumerate}	
\end{lemma}	

\begin{proof}
	\begin{enumerate}
		\item Let $u \in \Gamma^+$. Then $u \in (\overline{W})_f$ iff $u^\omega \in \overline{W}$ iff $u^\omega \notin W$ iff $u \notin W_f$ iff $u \in \Gamma^+ \setminus W_f$. (Note that $\varepsilon \notin (\overline{W})_f \cup (\Gamma^+ \setminus W_f)$.) 
		\item Let $u,v \in \Gamma^+$ be such that $uv \in (\overline{W})_f$. So $uv \notin W_f$ by Lemma~\ref{lem:na-Wf}.\ref{lem:na-Wf1}. By contraposition of the assumption, $vu \notin W_f$. Again by Lemma~\ref{lem:na-Wf}.\ref{lem:na-Wf1}, $vu  \in (\overline{W})_f$ (since $vu \in \Gamma^+$). This implication is an equivalence by symmetry. \qedhere
	\end{enumerate}	
\end{proof}

\begin{restatable}{lemma}{faBasic}\label{lem:fa-basic}
Consider $W\subseteq\Gamma^\omega$ prefix-independent.	 Let $K,L \subseteq \Gamma^*$. 
\begin{enumerate}
	\item\label{lem:fa-basic1}  $K L \subseteq W_f$ iff $L K \subseteq W_f$.
	\item\label{lem:fa-basic2}  $ K L \subseteq (\overline{W})_f$ iff $ L K \subseteq (\overline{W})_f$.	
\end{enumerate}
\end{restatable}

\begin{proof}
	\begin{enumerate}
		\item By symmetry it suffices to prove one implication. Let us assume that $K L \subseteq W_f$. Consider a word in $L K$, so the word can be written as $v u$ with $v \in L$ and $u \in K$. By assumption $u v \in W_f$, so $(u v)^\omega \in W$ by definition of $W_f$, so $u (v u)^\omega \in W$, so $(v u)^\omega \in W$ by prefix independence of $W$, so $v u \in W_f$. This shows that $L K \subseteq W_f$.
		\item The second result follows by applying Lemma~\ref{lem:na-Wf}.\ref{lem:na-Wf2} on Lemma~\ref{lem:fa-basic}.\ref{lem:fa-basic1}. \qedhere
		\end{enumerate}
\end{proof}

We use the following definition in the reminder of the section:

\begin{definition}
For $W \subseteq \Gamma^\omega$, if for all $a \in \Gamma$, we have $(W_{a,f})^\omega \subseteq W$ then we say that $W_{a,f}$ mixes well.
\end{definition}

\begin{restatable}{lemma}{propMix}\label{lem:pos-properties-mix}
	If $W \subseteq \Gamma^\omega$ is positionally determined on vertex-colored hub-cycle arenas, then $W_{a,f}$ mixes well. Symmetrically, for all $a \in \Gamma$, we have that $(\overline{W}_{a,f})$ mixes well.
\end{restatable}

\begin{proof}
Let $(u_n)_{n \in \N}$ be a family with members in $W_{a,f} = W_f \cap a\Gamma^*$. We want to prove that $u_0u_1u_2\dots \in W$. Let $(v_n)_{n \in \N}$ be a family such that $\forall n \in \N, u_n = av_n$.

Consider a hub-cycle game controlled by Adam with hub $a$. For each $n\in\mathbb{N}$, there is a cycle starting and ending at $a$ that traverses a sequence of states, resulting in $v_n$.

We first argue that Adam does not have any winning strategies: otherwise he would have a positional winning strategy. It would yield a play of the form $av_n = u_n$ for some $n$, contradiction since $u_n^\omega \in W$. Since Adam could yield the (non-positional) play $u_0u_1u_2\dots \in \Gamma^\omega$, it follows that $u_0u_1u_2\dots \in W$.

The second result holds directly since $\overline{W}$ also guarantees positional determinacy. 
Indeed, for any hub-cycle game $G$ and objective $W$, Adam has a winning strategy for $W$ in $G$ iff Eve has a winning strategy for $\overline{W}$ in the dual hub-cycle game $G'$ where Eve and Adam are inverted. It follows that  $\overline{W}$ is positionally determined on $G$ iff $W$ is positionally determined on $G'$, and this holds by assumption. 
\end{proof}

\begin{restatable}{lemma}{onePos}\label{lem:1p-pos}
Let $W \subseteq \Gamma^\omega$ be a prefix-independent objective that is positionally determined on vertex-colored hub-cycle arenas. Then, for any $a, b \in \Gamma$ and $L, L' \subseteq \Gamma^*$:
\begin{align*}
&\forall u \in L, \exists v \in L', aubv \in W_f \iff \exists v \in L', \forall u \in L, aubv \in W_f \\
&\forall u \in L, \exists v \in L', aubv \in \overline{W}_f \iff \exists v \in L', \forall u \in L, aubv \in \overline{W}_f
\end{align*}
\end{restatable}

The statement and the proof of this lemma are similar to the ones of Proposition~\ref{prop:1p-pos-edge}

\begin{proof}
Consider the first equivalence $\forall u \in L, \exists v \in L', aubv \in W_f \iff \exists v \in L', \forall u \in L, aubv \in W_f$. The second equivalence follows by symmetry on $\overline{W}$, which is also prefix-independent and also guarantees positional determinacy by inverting Eve and Adam.

The right-to-left implication clearly holds. For the left-to-right implication, we proceed by contradiction. Assume:
$\forall u \in L, \exists v \in L', aubv \in W_f \quad \text{and} \quad \forall v \in L', \exists u \in L, aubv \notin W_f$.

Take $a=b$ and construct two sequences $(u_i)_{i \geq 0} \in L^\mathbb{N}$ and $(v_i)_{i \geq 0} \in (L')^\mathbb{N}$ by induction:
\begin{itemize}
    \item Let $u_0 \in L$. By the first part of the assumption, $\exists v_0 \in L'$ such that $au_0av_0 \in W_f$.
    \item Given $v_i\in L'$, by the second part of the assumption, $\exists u_{i+1} \in L$ such that $au_{i+1}av_i \notin W_f$.
    \item Given $u_{i+1}\in L$, by the first part of the assumption, $\exists v_{i+1} \in L'$ such that $au_{i+1}av_{i+1} \in W_f$.
\end{itemize}

Consider the infinite word $w = au_0av_0au_1av_1au_2av_2 \dots$. Since $W$ is prefix-independent, $w \in W$ iff any of its suffixes is in $W$.
By construction, every block $au_iav_i$ is in $W_f$. Since $W$ is positionally determined in one-player arenas, Lemma~\ref{lem:pos-properties-mix} implies that the infinite concatenation of these winning blocks remains in the winning set:
$(au_0av_0)(au_1av_1)(au_2av_2) \dots \in W$

Now consider the word obtained by removing prefix $au_0$ from $w$: $w' = av_0au_1av_1au_2av_2 \dots$. By prefix-independence, $w \in W \iff w' \in W$. 
However,  for all $i \geq 0$, by prefix-independence and by Lemma~\ref{lem:fa-basic}.\ref{lem:fa-basic2}, $av_iau_{i+1} \in \overline{W}_f$ iff $au_{i+1}av_i \in \overline{W}_f$, which holds. Since $\overline{W}$ is also positionally determined in one-player arenas, applying Lemma~\ref{lem:pos-properties-mix} to these blocks gives
$(av_0au_1)(av_1au_2)(av_2au_3) \dots \in \overline{W}$.

This implies $w \in \overline{W}$ by prefix-independence. We have reached a contradiction where $w \in W$ and $w \in \overline{W}$. Therefore, the quantifier exchange must hold.
\end{proof}

\begin{lemma}\label{lem:fa-basic-bis}Let $W \subseteq \Gamma^\omega$ be a prefix-independent language that guarantees positional determinacy on all vertex-colored hub-cycle arenas. 
Let $a \in \Gamma$ and $K,L \subseteq W_{a,f}$. Then $K L \subseteq W_{a,f}$. 
\end{lemma}

\begin{proof}
Let $u \in K$ and $v \in L$. By Lemma~\ref{lem:pos-properties-mix}, $W_{a,f}$ mixes well and we have $(uv)^\omega \in W$, so $uv \in W_f$.
\end{proof}

\begin{lemma}\label{lem:absorb-AE}
Let $W \subseteq \Gamma^\omega$ be a prefix-independent language that is positionally determined on vertex-colored hub-cycle arenas. 
Let $u \in a\Gamma^*$ and $L \subseteq a\Gamma^*$ be such that $uL \subseteq W_f$. Then $\forall n \in \N\forall v \in L^{n+1}, u^{n+1}v \in W_f$.	
\end{lemma}

\begin{proof}
	We proceed by induction on $n$.
	\begin{itemize}
		\item Case $n=0$: let $v \in L$. Then $u^{n+1}v = uv \in W_f$ by assumption of the lemma. 
		\item Let us assume that the claim holds for $n-1 \geq 0$, and consider $v \in L^{n+1}$, so $v$ can be written $wv'$ with $w \in L$ and $v' \in L^n$. By IH, $u^nv' \in W_f$, so $v'u^n \in W_f$ by Lemma~\ref{lem:fa-basic}.\ref{lem:fa-basic1}. By assumption of the lemma, $uw \in W_f$, so $uwv'u^n \in W_f$ by Lemma~\ref{lem:fa-basic-bis}, so $ u^{n+1}v = uu^nwv' \in W_f$ by Lemma~\ref{lem:fa-basic}.\ref{lem:fa-basic1} again.	Therefore the claim holds for $n$.	 \qedhere
	\end{itemize}	
\end{proof}

\begin{lemma}\label{lem:assum-commute}
Let $W \subseteq \Gamma^\omega$ be a prefix-independent language that is positionally determined on vertex-colored hub-cycle arenas. Then the following holds: 
$\forall a,b \in \Gamma, K,L \subseteq \Gamma^*, (\forall v \in L, \exists u \in K, aubv \in W_f) \Rightarrow \exists u \in K, \forall v \in L, aubv \in W_f$.

\begin{proof}
	Let $a,b \in \Gamma$ and $K,L \subseteq \Gamma^*$ be such that $\forall v \in L, \exists u \in K, aubv \in W_f$. So $\forall v \in L, \exists u \in K, bvau \in W_f$ by Lemma~\ref{lem:fa-basic}.\ref{lem:fa-basic1}. So by Lemma~\ref{lem:1p-pos}, we have $\exists u \in K,\forall v \in L, bvau \in W_f$. Hence $\exists u \in K,\forall v \in L, aubv \in W_f$ by Lemma~\ref{lem:fa-basic}.\ref{lem:fa-basic1} again.
\end{proof}
\end{lemma}

\begin{lemma}\label{lem:absorb-EA}
Let $W \subseteq \Gamma^\omega$ be a prefix-independent language that is positionally determined on vertex-colored hub-cycle arenas.
Let $u \in a\Gamma^*$ and $L \subseteq a\Gamma^*$ be such that $uL \subseteq W_f$. Then $\exists n > 0,\ \forall v \in L^+, u^nv \in W_f$.	
\end{lemma}

\begin{proof}
Let $u' \in \Gamma^*$ be such that $u = au'$ and let $K := \{u'u^n \mid n \in \N\}$. Also let $L' \subseteq \Gamma^*$ be such that $L = aL'$, and let $L'' := L'L^*$. 

By Lemma~\ref{lem:absorb-AE} we have $\forall v \in L^+ \exists n > 0, u^nv \in W_f$, i.e. $\forall y \in L'' \exists x \in K, axay \in W_f$. So $\exists x \in K \forall y \in L'', axay \in W_f$ by Lemma~\ref{lem:assum-commute}. Said otherwise, $\exists n > 0 \forall v \in L^+, u^nv \in W_f$.
\end{proof}

\begin{restatable}{lemma}{absorbOne}\label{lem:absorb-one}
Assume that $W \subseteq \Gamma^\omega$ is positionally determined on vertex-colored hub-cycle arenas. 
Let $u \in W_{a,f}$ and $L \subseteq a\Gamma^*$ be such that $uL \subseteq W_f$. Then $uL^* \subseteq W_f$.
\end{restatable}

Our proof of this lemma involves a sequence of technical lemmas.

\begin{proof}
Let $E := \{n \geq 1 \mid u^nL^+ \subseteq W_f\}$, which is non-empty by Lemma~\ref{lem:absorb-EA}.
\begin{itemize}
	\item For all $n \in E$, we have $u^nL^+ \subseteq W_f$ by definition, so $uu^nL^+ \subseteq W_f$ since $u \in W_{a,f}$ and $W_{a,f}$ mixes well by Lemma~\ref{lem:pos-properties-mix}, so $n+1 \in E$.
	\item It suffices to show that $1 \in E$. Towards a contradiction, let us assume that $k := \min E > 1$. So let $v \in L^+$ be such that $u^{k-1}v \notin W_f$. 
    	\begin{itemize}
    		\item So $u^{k-1} v \in (\overline{W})_f$ by Lemma~\ref{lem:na-Wf}.\ref{lem:na-Wf1}, so $v u^{k-1} \in (\overline{W})_f$ by Lemma~\ref{lem:1p-pos}, so $vu^{k-1} u^{k-1}v \in (\overline{W})_f$ by Lemma~\ref{lem:pos-properties-mix}. So $u^{2k-2}vv \in (\overline{W})_f$.
    		\item Moreover, on the one hand $vv \in L^+$ since $v \in L^+$, and on the other hand $2 \leq k = \min E$, so $k \leq 2k - 2$, so $2k-2 \in E$, so $u^{2k-2}vv \in W_f$ by definition of $E$, contradiction.
    	\end{itemize}
So $\min E = 1$, i.e. $uL^+ \subseteq W_f$.
\end{itemize}		
Therefore $uL^* \subseteq W_f$ since $u \in W_f$.	
\end{proof}

We now build our generalized parity structure over the words in $\Gamma^*$.

Let $b \in \Gamma$. For all $u \in W_f$, let $g_b(u) := \{v \in (\overline{W})_{b,f}  \mid uv \in W_f\}$. For all $u,u' \in W_f$, let us write $u \sqsubseteq_b u'$ if $g_b(u) \subseteq g_b(u')$, and $u \sqsubset_b u'$ if $u \sqsubseteq_b u'$ and $g_b(u) \neq g_b(u')$, and $u \sim_b u'$ if $g_b(u) = g_b(u')$.

Let $a \in \Gamma$. Let $\sqsubseteq_{a,b}$, $\sqsubset_{a,b}$, $\sim_{a,b}$ be the restrictions of $\sqsubseteq_b$, $\sqsubset_b$, $\sim_b$, respectively, to $W_{a,f}$.

\begin{lemma}\label{lem:preord-eq}
Let $W \subseteq \Gamma^\omega$ be a prefix-independent objective that is positionally determined on vertex-colored hub-cycle arenas. For all $a,b\in\Gamma$, the following holds:
\begin{enumerate}
	\item\label{lem:preord-eq1} $\sqsubseteq_{a,b}$ is a total preorder, i.e. a binary relation that is reflexive, transitive, and total.
	\item\label{lem:preord-eq2} $\sim_{a,b}$ is an equivalence relation with finite index, i.e. finitely many equivalence classes.
\end{enumerate}
\end{lemma}

\begin{proof}
\begin{enumerate}
	\item $\sqsubseteq_{a,b}$ is a preorder since $\subseteq$ is a preorder. Let us show that it is total. Towards a contradiction, let $u,u' \in W_{a,f}$ be incomparable, i.e.  there exists $v \in g_b(u) \setminus g_b(u')$ and $v' \in g_b(u') \setminus g_b(u)$. So, on the one hand $uv, u'v' \in W_f$, and on the other hand $uv',u'v \notin W_f$. Thus, $vu,v'u' \in W_f$ by Lemma~\ref{lem:fa-basic}.\ref{lem:fa-basic1}, so $vuv'u' \in W_f$ since $W_{b,f}$ mixes well by Lemma~\ref{lem:pos-properties-mix}, so $uv'u'v$ by Lemma~\ref{lem:fa-basic}.\ref{lem:fa-basic2}. But $uv'u'v \notin W_f$ since $(\overline{W})_{a,f}$ mixes well, again by Lemma~\ref{lem:pos-properties-mix}, and we have a contradiction.
	\item By Lemma~\ref{lem:preord-eq}.\ref{lem:preord-eq1} the classes of $\sim$ are totally ordered by the quotient order. So it suffices to show that there are no infinite strictly monotone sequences for $\sqsubseteq_{a,b}$. Towards a contradiction, let $u_0 \sqsubset_{a,b} u_1 \sqsubset_{a,b} u_2 \dots$ be such a sequence. For all $n \in \N$, let $v_n \in g_b(u_{n+1}) \setminus g_b(u_n)$, i.e. $u_nv_n \notin W_f$ and $v_nu_{n+1} \in W_f$. So on the one hand $(u_0v_0)(u_1v_1)(u_2v_2)\dots \notin W$ by Lemma~\ref{lem:pos-properties-mix}, and on the other hand $(v_0u_1)(v_1u_2)(v_2u_3)\dots \in W$ again by Lemma~\ref{lem:pos-properties-mix}, contradicting prefix-independence. 
	
	In a similar way one can show that there are no infinite decreasing sequences. \qedhere
\end{enumerate}	
\end{proof}

\begin{lemma}\label{post-det-full-conseq}
Let $W \subseteq \Gamma^\omega$ be a prefix-independent objective that guarantees positional determinacy on all vertex-colored hub-cycle arenas. For all $a,b\in\Gamma$, there exists $n \in N$ and $B_1,\dots, B_{2n+1} \subseteq \Gamma^+$ such that :
\begin{enumerate}
	\item\label{Q9.1} The $B_{2i}$ form a partition of $W_{a,f}$, and $B_{2i} \neq \emptyset$ for all $i < n$..
	
	\item \label{Q9.1b}The $B_{2i+1}$ are pairwise disjoint, their union is $(\overline{W})_{b,f}$, and $B_{2i+1} \neq \emptyset$ for all $i < n$.
	
	\item\label{Q9.2} For all $i \in \{1,\dots,n\}$ we have $B_{2i} \subseteq W_f$.
	
	\item\label{Q9.3} For all $i \in \{0,\dots,n\}$ we have $B_{2i+1} \subseteq (\overline{W})_f$.
	
	\item\label{Q9.4} For all $i \in \{1,\dots,n\}$ and $k \leq 2i$ we have $B_{2i} B_{k} \subseteq W_f$.
	
	\item\label{Q9.5} For all $i \in \{0,\dots,n\}$ and $k \leq 2i+1$ we have $B_{2i+1} B_{k} \subseteq (\overline{W})_f$.
	
	\item\label{Q9.6} Let $(u_n)_{n\in \N}$ be a family of members in $W_{a,f} \cup \overline{W}_{b,f}$. For all $n \in \N$ let $p_{a,b}(n) \in \N$ such that $u_n \in B_{p_{a,b}(n)}$. Then $u_0u_1\dots \in W$ iff $\limsup_{n \in \N}p_{a,b}(n)$ is even. 
\end{enumerate}
\end{lemma}

	\begin{proof}
		Let us define the $B_i$ below.
		\begin{itemize}
			\item Let $B_2 \sqsubset_{a,b} B_4 \sqsubset_{a,b} \dots \sqsubset_{a,b} B_{2n}$ be the equivalence classes of $\sim_{a,b}$. They are finitely many by Lemma~\ref{lem:preord-eq}.\ref{lem:preord-eq2}.
			
			\item Let $B_1 :=g_b(B_2)$. (Note that $\forall u,v \in B_2,g_b(u) =g_b(v)$, hence the notation $g(B_2)$ used in the remainder.)
			
			\item For all $i \in \{1,\dots,n-1\}$ let $B_{2i+1} :=g_b(B_{2i+2}) \setminus g_b(B_{2i})$
			
			\item Let $B_{2n+1} := (\overline{W})_{b,f}\setminus g_b(B_{2n})$.
		\end{itemize}

		\begin{enumerate}
			\item The $B_{2i}$ form a partition of $W_{a,f}$ by definition of $\sim_{a,b}$. 
			
			\item By construction the $B_{2i+1}$ are pairwise disjoint and their union is $(\overline{W})_{b,f}$. For all $i < n$, each  $B_{2i+1} =g_b(B_{2i+2}) \setminus g_b(B_{2i})$, which is non-empty since $B_{2i+2}$ and $B_{2i}$ are two different equivalence classes.
			
			\item The domain of $\sim_{a,b}$ is $W_{a,f}$, so for all $i \in \{1,\dots,n\}$ we have $B_{2i} \subseteq W_{a,f}$.
			
			\item The co-domain of $g_b$ is $(\overline{W})_{b,f}$, so for all $i \in \{0,\dots,n\}$ we have $B_{2i+1} \subseteq (\overline{W})_{b,f}$.
			
			\item Let $i \in \{1,\dots,n\}$ and $k \leq 2i$ (so $k \leq 2n$). Let us show that $B_{2i}B_{k} \subseteq W_f$.
			\begin{itemize}
				\item If $k$ is even, $B_{k} \subseteq W_{a,f}$. Since $B_{2i} \subseteq W_{a,f}$ too, $B_{2i} B_{k} \subseteq W_{a,f}$ by Lemma~\ref{lem:fa-basic}.\ref{lem:fa-basic1}.
				
				\item If $k$ is odd, $k \leq 2n$ implies $B_{k} \subseteq g_b(B_{k+1}) \subseteq g_b(B_{2i})$ (since $k+1 \leq 2i$), so $B_{2i} B_{k} \subseteq W_f$ by definition of $g(B_{2i})$.
			\end{itemize}
			
			\item Let $i \in \{0,\dots,n\}$ and $k \leq 2i+1$. Let us show that $B_{2i+1}B_{k} \subseteq (\overline{W})_f$.
			\begin{itemize}
				\item If $k$ is odd, $B_{k} \subseteq (\overline{W})_f$. Since $B_{2i+1} \subseteq (\overline{W})_f$ too, $B_{2i+1} B_{k} \subseteq (\overline{W})_f$ by Lemma~\ref{lem:fa-basic}.\ref{lem:fa-basic2}.
				
				\item If $k$ is even, then $k < 2i+1$ and $g(B_k) \subseteq g_b(B_{2i})$. So $B_{2i+1} \cap g_b(B_k) = \emptyset$, since $B_{2i+1} :=g_b(B_{2i+2}) \setminus g_b(B_{2i})$. Thus $B_{2i} B_{k} \subseteq (\overline{W})_f$ by definition of $g(B_{k})$.
			\end{itemize}
			\item 
Let $k = \limsup_{n \to \infty} p_{a,b}(n)$. By definition of the limit superior, there exists an index $N \in \mathbb{N}$ such that $p_{a,b}(n) \leq k$ for all $n \geq N$, and $p_{a,b}(n) = k$ for infinitely many $n \geq N$.

Because $W$ is prefix-independent, the finite prefix $u_0 \dots u_{N-1}$ does not affect membership in $W$. Let $i_0 < i_1 < i_2 < \dots$ be the strictly increasing sequence of indices $\geq N$ where the maximum priority is reached, meaning $p_{a,b}(i_j) = k$. 

We factor the infinite suffix $w = u_{i_0} u_{i_0+1} \dots$ into an infinite sequence of finite segments $w_0, w_1, w_2, \dots$ defined by $w_j = u_{i_j} u_{i_j+1} \dots u_{i_{j+1}-1}$

For each $j \in \mathbb{N}$, the segment $w_j$ consists of exactly one element from $B_k$ (which is $u_{i_j}$) followed by a finite, possibly empty, sequence of elements from blocks $B_m$ with $m < k$. Let $B_{<k} = \bigcup_{m \leq k} B_m$. Then $w_j \in B_k (B_{<k})^*$.

We have two cases, depending on the parity of $k$. Theses cases are symmetric, so we only present the case for even $k$ below.

If $k$ is even, let $k = 2i$. 
By Point \ref{Q9.4}, we know that $B_{2i} B_m \subseteq W_f$ for all $m \leq 2i$. Consequently, $B_{2i} L \subseteq W_f$. 
Applying Lemma~\ref{lem:absorb-one}, we extend this to $B_{2i} L^* \subseteq W_f$. 
Therefore, $w_j \in W_f$. Since $u_{i_j} \in B_{2i} \subseteq W_{a,f}$, $w_j$ starts with $a$, so $w_j \in W_{a,f}$. 
Because $W_{a,f}$ mixes well (by Lemma~\ref{lem:pos-properties-mix}), the infinite concatenation $w=w_0 w_1 w_2 \dots$ belongs to $W$. Thus, $u_0 u_1 \dots \in W$.

Hence $u_0 u_1 \dots \in W$ if and only if $\limsup_{n \to \infty} p_{a,b}(n)$ is even.  \qedhere
		\end{enumerate}	
		
	\end{proof}

We can now conclude, and introduce an additional result that we use later.

\begin{proof}[Proof of Theorem~\ref{thm:pose-parity-words}]
To obtain a generalized parity objective on anchored words as given in Definition~\ref{def:parity-words}, it suffices to take for every $a\in \Gamma$ the priority function $p_a = p_{a,a}$ built in Lemma~\ref{post-det-full-conseq}. 
\end{proof}

\begin{lemma}\label{lem:par-word-posdet}
Let $\Gamma$ be a non-empty set of colors, and $W \subseteq \Gamma^\omega$ be a generalized parity objective on anchored words. Then $W$ is positionally determined on vertex-colored hub-cycle arenas.
\end{lemma}

\begin{proof}
Let $G$ be a hub-cycle game where the vertices belong to Eve, with objective $W$, as in Definition~\ref{def:hub-cycle}. We assume that the vertices of the arena belong to Eve, in particular the only state with out-degree strictly greater than one, which we call the hub. The hub has some color $a\in \Gamma$. The case where it belongs to Adam can be done symetrically by replacing Adam and Eve in the remainder of the proof.

First, if Adam has a winning strategy, this strategy has to be positional since Adam does not control the hub, and it's the only vertex with more than one outgoing edge. 

Second, assume Eve has a winning strategy in $G$. Then there exists $w\in W$ resulting from a play on $G$. Since $G$ is strongly connected, it visits the hub colored with $a$ infinitely often, hence $w = a w_0 a w_1 a w_2\dots$ where the finite words $w_i$ are exactly the sequences that do not visit the hub, while every $a$ is seen at the hub. For $p_a$ the priority function on words associated to $a$, we take  $k = \limsup_{i \to \infty} p_a(a w_i)$. Then there is some $w_i$ such that $p_a(a w_i) = k$. The positional strategy playing $(a w_i)^\omega$ is winning.
\end{proof}

Although the following two results are not necessary for any of our proof, they provide additional context for the studied objectives.

\begin{restatable}{proposition}{threeOne}\label{prop:3red2}
Let $W$ be a generalized parity objective on anchored words.
	For all $a \in \Gamma$ and $u,v,w \in \Gamma^*$, if $auav, auaw \in W_f$ then $auavaw \in W_f$. Corollary (by symmetry): For all $a \in \Gamma$ and $u,v,w \in \Gamma^*$, if $auav, auaw \in\overline{W_f}$ then $auavaw \in \overline{W}_f$.
\end{restatable}

\begin{proof}		
	Since the objective $W$ (or more specifically its finitary version $W_f$) is stuctured as described in Lemma~\ref{post-det-full-conseq}, each of the three words $au$, $av$, $aw$ can be assigned a priority. If $auav, auaw \in W_f$, then the highest among these three priorities is even, which allows us to conclude.
\end{proof}

\begin{restatable}{corollary}{threeTwo}\label{corol:3red2}
Let $\Gamma$ be a finite non-empty set of colors, and $W\subseteq \Gamma^\omega$ be a generalized parity objective on anchored words.
 The objective $W$ is fully characterized by $W_f \cap \Gamma^{2|\Gamma|}$.
\end{restatable}

\begin{proof}
Let $x \in \Gamma^*$ be such that $2|\Gamma| < |x|$. Since $\Gamma$ is finite, there exists $a \in \Gamma$ that occurs at least three times in $x$. Then $x$ is the circular permutation of some $auavaw$. Proposition~\ref{prop:3red2} implies the following:
\begin{itemize}
    \item If two or three words among $auav$, $auaw$, $avaw$ are in $W_f$, then $auavaw \in W_f$.
    \item Else $auavaw \notin W_f$.
\end{itemize}
Note that $|auav|, |auaw|, |avaw| < |x|$. 
\end{proof}

\section{Vertex-colored games: parity on pairs}
\label{sec:vertex-pairs}

Let $\Gamma$ be a non-empty alphabet.  We have shown that all objectives $W\subseteq\Gamma^*$ positionally determined on vertex-colored hub-cycle one-player arenas are generalized parity objectives on anchored words. We now consider objectives on pairs, and introduce the following. 
For a finite word $u = a_0 a_1 \dots a_{n-1} \in \Gamma^+$, the \textbf{cyclic 2-factor set} is the set of consecutive pairs of colors: 
$F_2(u) = \{ (a_i, a_{i+1 \pmod n}) \mid 0 \leq i < n \}$.
For an infinite word $w \in \Gamma^\omega$, the \textbf{limit 2-factor set} $F_2^\infty(w)$ consists of those consecutive pairs of colors $(x, y) \in \Gamma^2$ such that $xy$ appears infinitely often as factor in $w$. Note that in both cases, pairs are ordered.
We say $W\subseteq\Gamma^\omega$ is a Muller objective on pairs if there exists $M\subseteq 2^{\Gamma^2}$ such that for all $w \in \Gamma^\omega$, we have $w \in W$ if and only if $F_2^\infty(w)\in M$. 
    
In this section, we will consider a finite $\Gamma$, and show \ref{sec2:2c}~$ \Rightarrow $ \ref{sec2:2d} in two steps. First, in Lemma~\ref{lem:muller-like}, we show that if $W \subseteq \Gamma^\omega$ is a generalized parity objective on anchored words, then $W$ satisfies a Muller objective on pairs. Second, using that generalized parity objectives on anchored words are positionally determined on vertex-colored hub-cycle arenas, as seen in Lemma~\ref{lem:par-word-posdet}, we then conclude with Lemma~\ref{lem:mul-is-par}, showing that $W$ is a generalized parity objective on pairs of colors.

Finally, for \ref{sec2:2d} $ \Rightarrow $ \ref{sec2:2b}, we show Theorem~\ref{thm:2parity-pos}: we consider any generalized parity objective on pairs. A vertex-colored game with this objective can be transformed into an edge-colored parity game by assigning to each edge $q_1 q_2$ the priority $p(a, b)$, where $a$ and $b$ are the original vertex colors. Since the resulting edge-colored game is a parity game, it is positionally determined, and its winning strategies transfer back to the original vertex-colored setting.

\begin{example}
We have shown in Example~\ref{ex:def-ab} that $W = (a+b)^*(ab)^\omega$ on $\Gamma=\{a,b\}$ is a generalized parity objective on pairs. We have shown in Example~\ref{ex:intro-ab} that it is positionally determined on the vertex-colored arena of Figure~\ref{subfig:vertex-col}. Proving \ref{sec2:2d} $ \Rightarrow $ \ref{sec2:2b} will tell us it is determined on all vertex-colored two-player arenas.
\end{example}

\begin{theorem}\label{thm:rev-case2}
    \begin{mdframed}[innertopmargin=+0.5em,innerbottommargin=+0.5em]
    Let $\Gamma$ be a finite non-empty set of colors. If $W \subseteq \Gamma^\omega$ is a generalized parity objective on anchored words, then $W$ is a generalized parity objective on pairs.
    \end{mdframed}
\end{theorem}

We introduce the following lemma to establish a closure property: rearranging an infinite word $w$ such that a finite word $u$ occurs infinitely often does not alter membership in $W$.

\begin{lemma}\label{lem:reorder-u}
Let $\Gamma$ be a finite non-empty set of colors. Let $W \subseteq \Gamma^\omega$ be a generalized parity objective on anchored words. Let $w \in \Gamma^\omega$ be an infinite word, and $u\in\Gamma^*$ such that $F_2(u)\subseteq F_2^\infty(w)$. Then there exists $w_u$ such that $F_2^\infty(w)=F_2^\infty(w_u)$, $u$ appears infinitely often in $w_u$, and $w\in W$ iff $w_u\in W$.
\end{lemma}

\begin{proof}
We proceed by considering every letter $u_i$ successively, and reordering words beginning with $u_i$ (without changing the acceptance condition given by Definition~\ref{def:parity-words}) until obtaining a word where $u=u_0 u_1 u_2\dots$ appears infinitely often.

Because all adjacent pairs $(u_j, u_{j+1})$ of the word $u$ belong to $F_2^\infty(w)$, they each occur infinitely often in $w$. We can therefore find a finite prefix $P_0$ of $w$, after which only pairs in $F_2^\infty(w)$ appear, and all do infinitely often. 

Hence the pair $(u_0, u_1)$ occurs after $P_0$. 
After this occurrence, since the pair $(u_1,u_2)$ also occurs infinitely often, some occurrence lies strictly later in the word. Iterating this argument for all consecutive pairs in $u$, we obtain a factorization of $w = w^0_{\mathrm{seg}} w^1_{\mathrm{seg}} w^2_{\mathrm{seg}}\dots$, where each $w^i_{\mathrm{seg}}$ is of the form
$w^i_{\mathrm{seg}} = P^i_0 \cdot u_0u_1 \cdot P^i_1 \cdot u_1u_2 \cdot P^i_2 \dots u_{n-1} u_n P^i_{n} \cdot u_nu_0 Q^i_0 u_1 Q^i_1 \dots u_n Q^i_n $. 
To preserve some readability, subscripts are treated modulo the length of the sequence, meaning that the index following the last element wraps around to 0. For instance, $u_{j+1}$ evaluates to $u_0$ when $j$ is the final index.
Each $P^i_j$ is the factor between the chosen occurrence of $u_{j-1}u_{j}$ and the next occurrence of $u_{j}u_{j+1}$, and $Q^i_j$ is the factor between the chosen occurrence of $u_{j}$ and the next occurrence of $u_{j+1}$. 

We highlight factors beginning with $u_1$ as follows: $w = (P^0_0 u_0)\cdot (u_1  P^0_1) \cdot (u_1u_2 P^0_2 u_2 u_3 P^0_3 \dots \\  u_{n-1} u_n P^0_{n} u_nu_0 Q^0_0) \cdot (u_1 Q^0_1 u_2 Q^0_2\dots u_n Q^0_n P^1_0 u_0)\cdot (u_1  P^1_1) \cdot (u_1u_2 P^1_2 u_2 u_3 P^1_3 \dots u_{n-1} u_n P^1_{n} u_nu_0 Q^1_0) \cdot u_1 Q^1_1 u_2 Q^1_2\dots u_n Q^1_n P^2_0 u_0\dots$

We now reorder those factors, defining  $w_1 = (P^0_0 u_0)\cdot (u_1u_2 P^0_2 u_2 u_3 P^0_3 \dots u_{n-1} u_n P^0_{n} u_nu_0 Q^0_0) \cdot (u_1  P^0_1) \cdot (u_1 Q^0_1 u_2 Q^0_2 \dots u_n Q^0_n P^1_0 u_0)\cdot (u_1u_2 P^1_2 u_2 u_3 P^1_3  \dots u_{n-1} u_n P^1_{n} u_nu_0 Q^1_0)\cdot (u_1 P^1_1)  \cdot u_1 Q^1_1 \\ u_2 Q^1_2\dots u_n Q^1_n P^2_0 u_0\dots$
Since $W$ is a generalized parity objective on anchored words, we consider $p_{u_{1}}$ as in Definition~\ref{def:parity-words}. Since the limit superior does not change by reordering, we have $w_1\in W$ iff $w\in W$. We remark that in $w_1$, every $P^i_0$ is followed by the sequence $ u_0 u_1 u_2$ 

We now do the same with words starting with $u_2$, with $w_1 = (P^0_0 u_0 u_1) \cdot (u_2 P^0_2) \cdot (u_2 u_3 P^0_3 \dots u_{n-1} u_n P^0_{n} u_nu_0 Q^0_0 u_1  P^0_1 u_1 Q^0_1) \cdot (u_2 Q^0_2 \dots u_n Q^0_n P^1_0 u_0 u_1) \cdot (u_2 P^1_2) \cdot (u_2 u_3 P^1_3  \dots \\ u_{n-1} u_n P^1_{n} u_nu_0 Q^1_0 u_1  P^1_1 u_1 Q^1_1)\cdot u_2 Q^1_2\dots u_n Q^1_n P^2_0 u_0\dots$

And reorder similarly , defining  $w_2  = (P^0_0 u_0 u_1)\cdot (u_2 u_3 P^0_3 \dots u_{n-1} u_n P^0_{n} u_nu_0 Q^0_0 u_1  P^0_1 u_1 Q^0_1)  \cdot (u_2 P^0_2)\cdot (u_2 Q^0_2 \dots u_n Q^0_n P^1_0 u_0 u_1) \cdot (u_2 u_3 P^1_3  \dots u_{n-1} u_n P^1_{n} u_nu_0 Q^1_0 u_1  P^1_1 u_1 Q^1_1) \cdot (u_2 P^1_2)\cdot u_2 Q^1_2\dots \\ u_n Q^1_n P^2_0 u_0\dots$ Considering $p_{u_{2}}$ ensures $w_2\in W$ iff $w_1\in W$, and we remark that in $w_2$, every $P^i_0$ is followed by the sequence $ u_0 u_1 u_2 u_3$.

We proceed this way for every letter of $u$, moving each factor $u_{j+1}P^i_j$ to the right of the corresponding prefix, until reaching $u_{n}$ and obtaining $w_n = P^0_0 \cdot (u_0u_1\dots u_n)\cdot (u_1 P^0_1 u_1 Q^0_1 u_2 P^0_2 u_2 Q^0_2 \dots u_n P^0_n u_n Q^0_n) \cdot P^0_0 \cdot (u_0u_1\dots u_n)\dots$ where every $P^i_0$ is followed by $u$. Hence we name $w_u = w_n$. We have that $w_u\in W$ iff $w\in W$, that $F_2^\infty(w)=F_2^\infty(w_u)$, and that $u$ appears infinitely often in $w_u$.
\end{proof}

\begin{lemma}\label{lem:muller-like}
 Let $\Gamma$ be a finite non-empty set of colors. If $W \subseteq \Gamma^\omega$ is a generalized parity objective on anchored words, then $W$ is a Muller objective on pairs. 
\end{lemma}

\begin{proof}
Assume that $W \subseteq \Gamma^\omega$ is a generalized parity objective on anchored words. 
Let $v, w \in \Gamma^\omega$, with $F_2^\infty(v) = F_2^\infty(w)$. Let $F=F_2^\infty(v)$. We show that $v \in W$ iff $w \in W$.

Let $a \in \Gamma$ be a color such that some pair $(a, b) \in F$. Since $\Gamma$ is finite and the word is infinite, such an $a$ exists and occurs infinitely often in both $v$ and $w$. By Definition~\ref{def:parity-words}, the acceptance of both words is determined by a priority function $p_a$.

Consider any ordering $<$ on $\Gamma$. We recall that the radix (or shortlex) total order $<_{rad}$, is defined as follows: For $u,v\in\Gamma^*$, we have $u <_{rad} v$ iff $|u| < |v|  \lor (|u| = |v| \land (\exists w, x, y \in \Gamma^*, a, b \in \Gamma : u=wax, v=wby, a < b))$.

We define a mapping $g: a\Gamma^* \to a\Gamma^*$ that maps every $u \in a\Gamma^*$ to the radix-minimal word $u'$ such that we have both $p_a(u) = p_a(u')$ and $F_2(u) = F_2(u')$.
Since the codomain of $p_a$ and the number of subsets of $\Gamma^2$ are both finite, the image $S_{a,F} = \{ g(u) \mid u \in a\Gamma^*, F_2(u) \subseteq F \}$ is a finite set of words.

Let $U$ be a finite word formed by concatenating all elements in $S_{a,F}$. We have $F_2(U) = F$ and apply  Lemma~\ref{lem:reorder-u}: we can reorder $v$ into $v_1$ such that $v \in W$ iff $v_1 \in W$, $F_2^\infty(v_1) = F$, and the block $U$ appears infinitely often in $v_1$. This forces every element of $S_{a,F}$ to appear infinitely often in $v_1$.

We factor $v_1$ into segments starting with $a$: $v_1 = x a u_0 a u_1 \dots$. We define $v_2 = x a g(u_0) a g(u_1) \dots$. By Definition~\ref{def:parity-words} of the generalized parity objective on anchored words, $v_1 \in W$ iff $v_2 \in W$ because $p_a(u_i) = p_a( g(u_i))$ for all $i$. Since $U$ occurs infinitely often in $v_1$, the set of segments $\{ a g(w_i) \}$ occurring infinitely often in $v_2$ is exactly $S_{a,F}$.

Applying the same procedure to $w$ yields a word $w_2$ such that $w \in W$ iff $w_2 \in W$. The set of segments occurring infinitely often in $w_2$ is also exactly $S_{a,F}$. 
Both $v_2$ and $w_2$ are infinite concatenations of segments where the set of priorities $\{ p_a(u) \mid u \text{ occurs infinitely often} \}$ is identical and equal to $\{ p_a(s) \mid s \in S_{a,F} \}$, and we have:
   
\noindent $\limsup_{i \to \infty} p_a(\text{segments of } v_2) = \max \{ p_a(s) \mid s \in S_{a,F} \} = \limsup_{i \to \infty} p_a(\text{segments of } w_2)$.
   
   By the generalized parity condition on words, $v_2 \in W$ iff $w_2 \in W$. It follows that $v \in W $ iff $w \in W$. The membership depends only on $F$, defining a Muller objective on pairs $M = \{ F \mid \max_{s \in S_{a,F}} p_a(s) \text{ is even} \}$.
\end{proof}

The proof of the following result is similar to the one in~\cite{Zielonka98}, but we slightly adapt the result to work with pairs of colors placed on vertices. We give it for the sake of completeness.

\begin{lemma}\label{lem:mul-is-par}
    Let $\Gamma$ be a finite alphabet. If $W \subseteq \Gamma^\omega$ is positional on all vertex-colored hub-cycle arenas, and $W$ is a Muller objective on pairs, then $W$ is a generalized parity objective on pairs.
\end{lemma}

\begin{proof}
Let $W \subseteq \Gamma^\omega$, and let $M \subseteq 2^{\Gamma^2}$ be the Muller objective on pairs such that $w \in W$ iff $F_2^\infty(w) \in M$. Let $U, V \in M$. Suppose there exists a letter $a \in \Gamma$ that appears infinitely often in both $U$ and $V$. Specifically, let $(a, b) \in U$ and $(a, c) \in V$. 

Since $U \in M$, there exists a word $u = (ab u_2 \dots u_n)^\omega \in W$. Since $V \in M$, there also exists a word $v = (ac v_2 \dots u_m)^\omega \in W$.

We consider the one-player game starting from an initial state $a$, belonging to player $P_2$, from where two actions are possible: the first action leads to a chain that reads $(ab u_2 \dots u_n)$ before returning to the initial state, the second action leads to a chain that reads $(ac v_2 \dots u_m)$ before returning to the initial state. Positional strategies will only yield $u$ or $v$, which are both in $W$ and thus winning for $P_1$, hence $P_2$ has no winning strategies. This means that alternating between the two actions is also winning for $P_1$, and so the word obtained this way $w= ab u_2 \dots u_n \cdot ac v_2 \dots v_m$ is in $W$. Since $F_2^\infty(w) = U\cup V$, we have $U\cup V\in M$.

Thus, $M$ is closed under the union of sets that share a common letter in at least one pair. 

We construct the priority function $p: \Gamma^2 \to \mathbb{N}$. To do so, let $F \subseteq 2^{\Gamma^2}$ be defined as $F = \{f\subseteq {\Gamma^2}\ | \ \exists w \in \Gamma^\omega, F_2^\infty(w) = f\}$. We first define the tree $Z$ associated with $M$ and $F$. A node in $Z$ is a set $S \in F$ by constructing it recursively:
\begin{itemize}
    \item The root of $Z$ is $\Gamma^2$.
    \item For any internal node $S \in Z$:
    \begin{itemize}
        \item If $S \in M$, its children are the maximal sets $C \subsetneq S$ such that $C \in F \setminus M$.
        \item If $S \notin M$, its children are the maximal sets $C \subsetneq S$ such that $C \in   \cap M$.
    \end{itemize}
\end{itemize}
Since $W$ is positional, the children $C_1,C_2\in M$ (resp. $\notin M$) of any node $S$ do not share any $a\in \Gamma$ such that $(a,b)\notin C_1$ (resp. $\in C_1$) and $(a,c)\notin C_2$ (resp. $\in C_2$), otherwise their union $C_1 \cup C_2 \notin M$ (resp. $\in M$), which contradicts their maximality. 
 Hence, all children are strictly disjoint.

Because the children are disjoint, every pair $e \in \Gamma^2$ belongs to a unique minimal node $S_e$ in the tree $Z$ (by subset inclusion). 

We define the priority function $p: \Gamma^2 \to \mathbb{N}$ recursively from the leaves to the roots to satisfy a max-generalized parity condition. To do so, we first define $p$ on $F$
\begin{itemize}
    \item If $S$ is a leaf, $p(S) = 0$ if $S \in M$, and $p(S) = 1$ if $S \notin M$.
    \item If $S \in M$ is an internal node, $p(S)$ is the smallest even integer strictly greater than $\max \{ p(C) \mid C \text{ is a child of } S \}$.
    \item If $S \notin M$ is an internal node, $p(S)$ is the smallest odd integer strictly greater than $\max \{ p(C) \mid C \text{ is a child of } S \}$.
\end{itemize}
For each pair $e \in \Gamma^2$, we set the transition priority $p(e) = p(S_e)$.

Consider any word $w \in \Gamma^\omega$ and let $U = F_2^\infty(w)$. Let $S_U$ be the unique minimal node in $Z$ such that $U \subseteq S_U$. 
Since $U$ is not contained in any single child of $S_U$, $U$ must contain at least one edge $e$ that belongs to $S_U$ but to no child of $S_U$. For this edge, the minimal containing node is $S_e = S_U$, so $p(e) = p(S_U)$. 

For all other edges $e' \in U$, their minimal containing node is either $S_U$ or a descendant of $S_U$. By construction, the priority of a node is strictly greater than the priorities of its descendants. Therefore, the maximum priority among all edges in $U$ is exactly $p(S_U)$. 

Since $w$ traverses all edges in $U$ infinitely often, $\limsup_{i \to \infty} p(w_i, w_{i+1}) = p(S_U)$. 
By the definition of $p(S_U)$, this maximum priority is even if and only if $S_U \in M$. 

Finally, because $U$ is contained in $S_U$ and is not contained in any child of $S_U$, $U$ must have the same membership status in $M$ as $S_U$. If $S_U \in M$ and $U \notin M$, $U$ would be an element of $F$ not in $M$, meaning it would be contained in one of the children of $S_U$, contradicting the minimality of $S_U$. The symmetric argument holds if $S_U \notin M$. 
Thus, $U \in M$ iff $S_U \in M$. 

We conclude that $w \in W$ iff $F_2^\infty(w) \in M$ iff $\limsup_{i \to \infty} p(w_i, w_{i+1}) \equiv 0 \pmod 2$, and so $W$ is a generalized parity objective on pairs.
\end{proof}

We can now give the following:

\begin{proof}[Proof of Theorem~\ref{thm:rev-case2}]
Let $W \subseteq \Gamma^\omega$ be a generalized parity objective on anchored words. Lemma~\ref{lem:muller-like} shows that $W$ is a Muller objective on pairs, and since $W$ is positionally determined on vertex-colored hub-cycle arenas by Lemma~\ref{lem:par-word-posdet}, we use Lemma~\ref{lem:mul-is-par} to conclude that $W$ is a generalized parity objective on pairs.
\end{proof}

\begin{theorem}\label{thm:2parity-pos}
    \begin{mdframed}[innertopmargin=+0.5em,innerbottommargin=+0.5em]
    Let $W \subseteq \Gamma^\omega$ be a generalized parity objective on pairs. Then $W$ is prefix-independent and positionally determined on all vertex-colored two-player arenas.
    \end{mdframed}
\end{theorem}

\begin{proof}
First, $W \subseteq \Gamma^\omega$ is prefix-independent because it is defined by a $\limsup$. Since the limit superior of a sequence is unaffected by finite prefixes, any change to a finite prefix of a word leaves its membership in $W$ unchanged.

To show positional determinacy, let $G = \langle A, W \rangle$ be a vertex-colored game on arena $A = \langle V_E, V_A, E \rangle$. Since $W$ is a generalized parity objective on pairs, there exists a priority function $p: \Gamma^2 \to \{0, \dots, n\}$ such that $w_0 w_1 \dots \in W$ iff $\limsup_{i \to \infty} p(w_i, w_{i+1})$ is even.

We consider the auxiliary alphabet $\Gamma^2$ and define $W_{\text{pair}} \subseteq (\Gamma^2)^\omega$ such that a word $u = (x_0,y_0)(x_1,y_1)\dots \in W_{\text{pair}}$ iff $\limsup_{i \to \infty} p(x_i, y_i)$ is even. By Definition~\ref{def:parity}, this makes $W_{\text{pair}}$ a generalized parity objective over $\Gamma^2$.

We construct the auxiliary edge-colored game $G_{\text{pair}} = \langle A_{\text{pair}}, W_{\text{pair}} \rangle$ where $A_{\text{pair}} = \langle V_E, V_A, E_{\text{pair}} \rangle$. We define the edges such that $(v,(a,b),v')\in E_{\text{pair}}$ iff $(v,a,v') \in E$ and $b$ is the unique element of $\Gamma$ such that for some $v''\in V$, we have $(v', b, v'')\in E$.

By construction of $E_{\text{pair}}$, any valid play in $G_{\text{pair}}$ produces a trace $u = (w_0,w_1)(w_1,w_2)\dots \in (\Gamma^2)^\omega$. This trace belongs to $W_{\text{pair}}$ iff $\limsup_{i \to \infty} p(w_i, w_{i+1})$ is even, which holds iff the underlying trace $w_0 w_1 \dots$ belongs to $W$.

By Theorem~6 of~\cite{Zielonka98}, $W_{\text{pair}}$ is positionally determined on all two-player edge-colored arenas. Hence, there exists a winning strategy $\sigma: V_E \to V$ on $G_{\text{pair}}$ that depends only on the current vertex. Because $G$ and $G_{\text{pair}}$ share the same vertices and their winning conditions are equivalent over all valid plays, $\sigma$ is also a positional winning strategy for $G$. Thus, $W$ is positionally determined on all vertex-colored two-player arenas.
\end{proof}

\section{Additional remarks}
\label{sec:remarks}

\textbf{$\omega$-regularity: } When $\Gamma$ is finite, the prefix-independent objectives characterizing positionally determined vertex-colored arenas are $\omega$-regular. 

\begin{restatable}{lemma}{omeWord}\label{lem:omeWord}
Let $\Gamma$ be a finite non-empty set of colors. Languages $W\subseteq\Gamma^\omega$ defined by a generalized parity objective on pairs are $\omega$-regular.
\end{restatable}

\begin{proof}
Let $W \subseteq \Gamma^\omega$ be a language defined by a generalized parity objective on pairs. By definition, there exists a priority function $p: \Gamma \times \Gamma \to [0,n]$ such that a word $w = w_0 w_1 w_2 \dots \in \Gamma^\omega$ belongs to $W$ iff the largest integer occurring infinitely often in the sequence $p(w_0, w_1), p(w_1, w_2), \dots$ is even.

For $a,b\in \Gamma$, let $M_{a,b} = \{w \in \Gamma^*\ | k+1=|w|,\ w_0=a,\ w_1=b,\ \forall i\in[0,k-1],\  p(w_i,w_{i+1}\leq p(a,b)$ and $p(w_{k},a)\leq p(a,b) \}$ contain all words starting with $ab$ and with priorities always less or equal to $p(a,b)$. This language is regular: it is recognized by a finite automaton with states $\Gamma \cup \{f,a_{init},b_{next}\}$, where $f$ is the only accepting state. From the initial state $a_{init}$, the first transition must reach $b_{next}$. Subsequent transitions between any $\gamma, \gamma' \in \Gamma$ are allowed iff $p(\gamma, \gamma') \leq p(a,b)$. Finally, a transition to $f$ is permitted from any $\gamma$ such that $p(\gamma, a) \leq p(a,b)$.

Since every word $w$ must contain infinitely often some pair $ab$ such that $p(a,b)$ is the maximal priority seen infinitely often in $w$, we then take:

\begin{equation}
    W = \Gamma^* \cdot \bigcup_{\substack{i \in [0,n] \\ i \text{ even}}} \bigcup_{\substack{a,b \in \Gamma \\ p(a,b)=i}} (M_{a,b})^\omega
\end{equation}

Since these languages are obtained by binary concatenation, finite union, and $\omega$-power of regular languages, they are $\omega$-regular. 
\end{proof}

\noindent\textbf{Infinite alphabet: } We have proved our characterization with generalized parity objectives on pairs only when considering a finite sets of colors. For an arbitrary set of colors, this characterization does not hold any more. The proof relies on parity games with infinite colors defined in~\cite{GW06} by Gr{\"{a}}del and Walukiewicz to create a chain of strictly increasing priorities of arbitrary length.

\begin{restatable}{lemma}{noInfinite}\label{lem:noInfinite}
Let $\Gamma = \mathbb{N}$. There exists an  objective $W\subseteq \Gamma^\omega$ prefix-independent and positionally determined on all vertex-colored two-player arenas, such that for every generalized parity function on pairs $W_p$, there exists a word $w \in \Gamma^\omega$ such that $w \in W \iff w \notin W_p$.
\end{restatable}

\begin{proof}
We define the objective $W$ based on an extension of the min generalized parity condition to an infinite set of priorities.  $u = u_0 u_1 \dots \in W \subseteq \Gamma^\omega$ iff $\liminf_{i \to \infty} u_i \neq \emptyset \Rightarrow \liminf_{i \to \infty} u_i $ is even.
Hence, a word is in $W$ if either the set of priorities appearing infinitely often is empty, or otherwise if its minimum element is even. 

$W$ is prefix-independent, since the limit inferior is independent of any finite prefix of $w$. Positional determinacy on all vertex-colored two-player arenas is non-trivial, and has been shown in~\cite{GW06}.

Let  $W_p$ be a generalized parity function on pairs of order $n$ (hence it has n priorities in $[0,n]$). To show it does not captures $W$, we consider for each $k \in [0,n+1]$ the word $w_k = (k, k+1, \dots, n+1, n+1, n, \dots, k+1)^\omega$. The set of elements of $\Gamma^2$ traversed infinitely often by $w_k$ is $F_2^\infty(w_n) = \{ (i, i+1), (i+1, i) \mid k \leq i \leq n\}\cup \{(n+1,n+1)\}$. Note that for all $n$, $F_2^\infty(w_{k+1}) \subsetneq F_2^\infty(w_{k})$.

By definition of $W$, $w_k \in W$ if and only if $k$ is even. For the generalized parity condition on pairs $W_p$, $w_k \in W_p$ iff the priority $M_n=\max_{k \leq i \leq n}(p(w_i,w_{i+1}))$ is even.

Since for all $k\in\mathbb{N}$ we have $F_2^\infty(w_{k+1}) \subsetneq F_2^\infty(w_{k})$, we have $M_{k+1} \leq M_{k}$.
If $W_p$ captures $W$, then for every $k$, $M_k$ must have the same generalized parity as $k$, hence $M_k \not\equiv M_{k+1} \pmod 2$.
These two conditions imply that $M_{k+1} < M_{k}$ for all $n \in [0,n+1]$.

The sequence of priorities $(M_k)_{k \in [0,n+1]}$ must be strictly decreasing. This contradicts the assumption that the range of $p$ is $[0,n]$. Hence, no such generalized parity function on pairs $W_p$ can represent $W$.
\end{proof}

\subsection{Games on finite arenas}
\label{app:notMon}

We compare our setting with Theorem 2 of~\cite{GZ05}, that considers quantitative objectives represented by a binary relation $\sqsubseteq$. It proves that $\sqsubseteq$ is positionally determined on \textbf{finite} edge-colored one-player games iff $\sqsubseteq$ is positionally determined on finite edge-colored two-player games iff $\sqsubseteq$ is \textbf{selective and monotone}. 

Prefix-independence implies monotonicity, but our running example $(a+b)^*(ab)^\omega$ on $\Gamma=\{a,b\}$ is not selective.

Before showing that there exists a non-selective objective that is positionally determined on vertex-colored games, we take the time to introduce the results of~\cite{GZ05}, restated in our framework.

\begin{definition}[Outcome games and strategies]\label{def:outcome}
An \textbf{outcome game} is a tuple $G = \langle A, \sqsubseteq \rangle$ where $A \in \mathcal{A}_e$ is an arena and Eve has a \textbf{preference relation} $\sqsubseteq$, which is a complete, reflexive, and transitive binary relation over $\Gamma^\omega$
We write $x \sqsubset y$ (strictly preferred) if $x \, \sqsubseteq \, y$ holds but $y \, \sqsubseteq \, x$ does not.

Plays and strategies are defined as in classical games.
A pair of strategies $(\sigma_{E}, \sigma_{A})$ for Eve and Adam is \textbf{optimal} if for all states $v \in V$ and all strategies $(\tau_{E}, \tau_{A})$ for respectively Eve and Adam: 
$\text{color}(p_G(s, \tau_{E}, \sigma_{A})) \,\, \sqsubseteq \,\, \text{color}(p_G(s, \sigma_{E}, \sigma_{A})) \,\, \sqsubseteq \,\, \text{color}(p_G(s, \sigma_{E}, \tau_{A}))$.

A strategy is an \emph{optimal positional strategy} if it is an positional strategy that satisfies the optimality condition.
\end{definition}

\begin{definition}[Recognizable Languages]
Let $Rec(C)$ denote the family of \textbf{recognizable subsets} of finite words over $C$. A language $L \subseteq C^*$ is recognizable if it is recognized by a finite automaton; this is equivalent to the class of \textbf{regular languages}
We also denote by $Pref(L)$ the set of prefixes of $L$. 

The \textbf{operator $[\cdot]$} associates a language of finite words $L \subseteq C^*$ with a set of infinite words $[L] \subseteq C^\omega$, with 
$ [L] = \{ x \in C^\omega \mid \text{every finite prefix of } x \text{ is in } Pref(L) \} $.
\end{definition}

\begin{definition}[Preferences over Sets]
The \textbf{preference relation} $\sqsubseteq$ and its strict version $\sqsubset$ are extended from individual sequences to \textbf{sets of sequences} $X, Y \subseteq C^\omega$ as follows:
\begin{itemize}
    \item $X \, \sqsubseteq \, Y$ iff $\forall x \in X, \exists y \in Y$ such that $x \, \sqsubseteq \, y$. (For every outcome in $X$, there is an outcome in $Y$ at least as good.)
    \item $X \sqsubset Y$ iff $\exists y \in Y, \forall x \in X$ such that $x \sqsubset y$. (There exists an outcome in $Y$ that is strictly better than every outcome in $X$.)
\end{itemize}
\end{definition}

\begin{definition}[Monotone and Selective relations]
A preference relation $v$ is:
\begin{itemize}
    \item \textbf{Monotone} if for all $M, N \in Rec(C)$:
    \[ \exists x \in C^*, [xM] \sqsubset [xN] \Rightarrow \forall y \in C^*, [yM] \, \sqsubseteq \, [yN] \]
    This implies that the optimal choice between two future behaviors does not depend on the history of the play. This is stronger than prefix-independent, which requires that the value of a play does not depend on the history.
    \item \textbf{Selective} if for all $x \in C^*$ and all $M, N, K \in Rec(C)$:
    \[ [x(M \cup N)^*K] \, \sqsubseteq \, [xM^*] \cup [xN^*] \cup [xK] \]
    This implies that a player cannot change the optimal choice by switching back and forth between different behaviors ($M$ and $N$).
\end{itemize}
\end{definition}

Hence, we can properly state Theorem 2 of~\cite{GZ05}.

\begin{proposition}[By Gimbert and Zielonka in~\cite{GZ05}]
Given a preference relation $\sqsubseteq$, both players have optimal positional
strategies for all edge-colored one-player outcome arenas  $G = \langle A, \sqsubseteq \rangle$ over finite arenas $A$ if and only if the relations
$\sqsubseteq$ and its inverse $\sqsubseteq^{-1}$ are monotone and selective.
\end{proposition}

We can now give our Lemma~\ref{lem:notMon}.

\begin{restatable}{lemma}{notMon}\label{lem:notMon}
There exists $W\subseteq \Gamma^\omega$ prefix-independent and positionally determined on all vertex-colored one-player arenas, such that its associated relation $\sqsubseteq$ is not selective.
\end{restatable}

\begin{proof}
Let $W = (a+b)^*(ab)^\omega$ on $\Gamma=\{a,b\}$, which is the generalized parity condition such that $p(a,a)=p(b,b)=1$ and $p(a,b)=p(b,a)=0$, hence by Theorem~\ref{thm:rev-case2} it is prefix-independent and positionally determined on all vertex-colored two-player arenas.
$W$ is monotone: prefix-independence ($x \in W \iff ax \in W$) ensures $x \sqsubseteq y \Rightarrow ax \sqsubseteq ay$.

However, $W$ is not selective: 
$W \cap \{a\}^\omega = \emptyset, W \cap \{b\}^\omega = \emptyset, \text{ yet } W \cap \{a, b\}^\omega \neq \emptyset$
shows that $W$ is not closed under the union of sub-alphabets.
\end{proof}

\section{Related works and future directions}
\label{sec:rel}

In 2005, Gimbert and Zielonka~\cite{GZ05} established a foundational  $1$-to-$2$-player lift for quantitative objectives represented by a preference relation $\sqsubseteq$. They proved that $\sqsubseteq$ is positionally determined on finite one-player edge-colored arenas iff it is positionally determined on finite two-player edge-colored arenas iff $\sqsubseteq$ is selective and monotone. Here, quantitative objectives assign non-binary payoffs to plays, generalizing standard qualitative win-lose objectives. This raises an \emph{open question} for our framework: Is $\sqsubseteq$ positionally determined on finite one-player vertex-colored arenas iff it is positionally determined on finite two-player vertex-colored arenas iff $\sqsubseteq$ is selective and monotone on pairs?

Finite-memory strategies offer a natural generalization of positional determinacy while remaining finitely representable, retaining some simplicity in terms of description and implementation. Characterizations of finite-memory determinacy extending the work of~\cite{GZ05} have been obtained for finite games~\cite{BLORV22} and infinite games~\cite{BRV23}. This raises the \emph{open question} whether our pair-based framework can bridge edge and vertex-colored games for finite-memory determinacy.

Determinacy that is positional or finite-memory for one player only (called half-positional determinacy, etc) has also been extensively studied~\cite{Kopczynski06, DBLP:journals/theoretics/Ohlmann23, CO26}. In this framework, our notion of objectives on pairs might be adapted to establish a correspondence between edge-colored and vertex-colored games.

A recent work by Colcombet and Idir~\cite{CI26} considers the $\omega$-regular objectives that are Eve-positional (Eve has a positional winning strategy whenever she has a winning strategy). For edge-colored finite games, they give a characterization of this property.

Finally, lifting the finiteness assumption in our second equivalence remains an open problem. As discussed previously, resolving this is particularly relevant for capturing parity games with infinitely many priorities~\cite{GW06}.

\section{Conclusion}
\label{sec:conclusion}

\textbf{Further remarks: } The finiteness of $\Gamma$ is used only to obtain the Muller objective on pairs in Lemma~\ref{lem:mul-is-par}: since the first step of the proof, Theorem~\ref{thm:pose-parity-words}, with a generalized parity objective on anchored words, remains valid for infinite color sets, we emphasize this intermediate result in the body of the article. Furthermore, the finiteness assumption in the second step is essential: parity games on infinitely many natural numbers are positionally determined~\cite{GW06}, yet they generally cannot be reduced to parity objectives on pairs of colors from consecutive vertices, as shown in Lemma~\ref{lem:noInfinite}. Finally, we establish in Section~\ref{sec:remarks} that our characterization with generalized parity objectives on pairs of colors from consecutive vertices is $\omega$-regular.

\textbf{Discussion: }
Our results establish a $1$-to-$2$-player lift, as the characterization of determinacy for two-player games reduces to the study of one-player hub-cycle games. This starts clarifying the relationship between edge and vertex games. Although any vertex-colored game can be viewed as an edge-colored game via a direct injection, the class of edge-colored games is strictly richer. This is reflected in the objectives: generalized parity objectives on a single color form a strict subclass of generalized parity objectives on pairs of colors. For instance, the objective $(a+b)^*(a^\omega + b^\omega)$ can be defined as an objective on pairs, but not on a single color.

Positionally determined objectives on vertex-colored arenas are strictly more expressive than those on edge-colored arenas. Indeed, vertex-colored games allow more complex objectives to remain positionally determined by effectively using the vertex color to track transitions. While we can convert a vertex-colored game with an objective on pairs into a standard edge-colored game, as in Theorem~\ref{thm:2parity-pos}, this incurs a quadratic blowup in the color domain, since the priority function has domain $\Gamma$ on single colors, but $\Gamma^2$ on pairs of colors.

Our equivalence thus establishes a two-way correspondence, between the vertex setting and the edge setting from~\cite{CN06,Kopczynski08}, via generalized parity objectives. Future works could strengthen this result.

\newpage

\bibliography{ref}

@phdthesis{Kopczynski08,
  title        = {Half-positional Determinacy of Infinite Games.},
  author       = {Eryk Kopczy\'{n}ski},
  year         = 2018,
  url         = {https://www.mimuw.edu.pl/~erykk/papers/hpwc.pdf},
  school       = {University of Warsaw},
  type         = {Ph{D} thesis}
}

@inproceedings{EJ91,
  author       = {E. Allen Emerson and
                  Charanjit S. Jutla},
  title        = {Tree Automata, Mu-Calculus and Determinacy (Extended Abstract)},
  booktitle    = {32nd Annual Symposium on Foundations of Computer Science, San Juan,
                  Puerto Rico, October 1-4, 1991},
  pages        = {368--377},
  publisher    = {{IEEE} Computer Society},
  year         = {1991},
  url          = {https://doi.org/10.1109/SFCS.1991.185392},
  doi          = {10.1109/SFCS.1991.185392},
  bibsource    = {dblp computer science bibliography, https://dblp.org}
}

@misc{Mostowski91,
 author       = {A. Mostowski},
  title        = {Games with forbidden positions},
  booktitle    = {Research report 78},
  year = {1991},
  publisher = {University of Gdansk}
 }

@article{Zielonka98,
  author       = {Wies\l{}aw Zielonka},
  title        = {Infinite Games on Finitely Coloured Graphs with Applications to Automata
                  on Infinite Trees},
  journal      = {Theor. Comput. Sci.},
  volume       = {200},
  number       = {1-2},
  pages        = {135--183},
  year         = {1998},
  url          = {https://doi.org/10.1016/S0304-3975(98)00009-7},
  doi          = {10.1016/S0304-3975(98)00009-7},
  bibsource    = {dblp computer science bibliography, https://dblp.org}
}

@article{CN06,
  author       = {Thomas Colcombet and
                  Damian Niwinski},
  title        = {On the positional determinacy of edge-labeled games},
  journal      = {Theor. Comput. Sci.},
  volume       = {352},
  number       = {1-3},
  pages        = {190--196},
  year         = {2006},
  url          = {https://doi.org/10.1016/j.tcs.2005.10.046},
  doi          = {10.1016/J.TCS.2005.10.046},
  bibsource    = {dblp computer science bibliography, https://dblp.org}
}

@inproceedings{GZ05,
  author       = {Hugo Gimbert and
                  Wies\l{}aw Zielonka},
  editor       = {Mart{\'{\i}}n Abadi and
                  Luca de Alfaro},
  title        = {Games Where You Can Play Optimally Without Any Memory},
  booktitle    = {{CONCUR} 2005 - Concurrency Theory, 16th International Conference,
                  {CONCUR} 2005, San Francisco, CA, USA, August 23-26, 2005, Proceedings},
  series       = {Lecture Notes in Computer Science},
  pages        = {428--442},
  publisher    = {Springer},
  year         = {2005},
  url          = {https://doi.org/10.1007/11539452\_33},
  doi          = {10.1007/11539452\_33},
  bibsource    = {dblp computer science bibliography, https://dblp.org}
}

@article{BLORV22,
  author       = {Patricia Bouyer and
                  St{\'{e}}phane Le Roux and
                  Youssouf Oualhadj and
                  Mickael Randour and
                  Pierre Vandenhove},
  title        = {Games Where You Can Play Optimally with Arena-Independent Finite Memory},
  journal      = {Log. Methods Comput. Sci.},
  volume       = {18},
  number       = {1},
  year         = {2022},
  url          = {https://doi.org/10.46298/lmcs-18(1:11)2022},
  doi          = {10.46298/LMCS-18(1:11)2022},
  bibsource    = {dblp computer science bibliography, https://dblp.org}
}

@article{BRV23,
  author       = {Patricia Bouyer and
                  Mickael Randour and
                  Pierre Vandenhove},
  title        = {Characterizing Omega-Regularity through Finite-Memory Determinacy
                  of Games on Infinite Graphs},
  journal      = {TheoretiCS},
  volume       = {2},
  eid          = {1},
  year         = {2023},
  url          = {https://doi.org/10.46298/theoretics.23.1},
  doi          = {10.46298/THEORETICS.23.1},
  bibsource    = {dblp computer science bibliography, https://dblp.org}
}

@inproceedings{Kopczynski06,
  author       = {Eryk Kopczy\'{n}ski},
  editor       = {Michele Bugliesi and
                  Bart Preneel and
                  Vladimiro Sassone and
                  Ingo Wegener},
  title        = {Half-Positional Determinacy of Infinite Games},
  booktitle    = {Automata, Languages and Programming, 33rd International Colloquium,
                  {ICALP} 2006, Venice, Italy, July 10-14, 2006, Proceedings, Part {II}},
  series       = {Lecture Notes in Computer Science},
  pages        = {336--347},
  publisher    = {Springer},
  year         = {2006},
  url          = {https://doi.org/10.1007/11787006\_29},
  doi          = {10.1007/11787006\_29},
  bibsource    = {dblp computer science bibliography, https://dblp.org}
}

@article{DBLP:journals/theoretics/Ohlmann23,
  author       = {Pierre Ohlmann},
  title        = {Characterizing Positionality in Games of Infinite Duration over Infinite
                  Graphs},
  journal      = {TheoretiCS},
  volume       = {2},
  eid          = {3},
  year         = {2023},
  url          = {https://doi.org/10.46298/theoretics.23.3},
  doi          = {10.46298/THEORETICS.23.3},
  bibsource    = {dblp computer science bibliography, https://dblp.org}
}

@article{CO26,
  author       = {Antonio Casares and
                  Pierre Ohlmann},
  title        = {Positional {\(\omega\)}-regular languages},
  journal      = {TheoretiCS},
  volume       = {5},
  eid          = {5},
  year         = {2026},
  url          = {https://doi.org/10.46298/theoretics.26.5},
  doi          = {10.46298/THEORETICS.26.5},
  bibsource    = {dblp computer science bibliography, https://dblp.org}
}

@article{GW06,
  author       = {Erich Gr{\"{a}}del and
                  Igor Walukiewicz},
  title        = {Positional Determinacy of Games with Infinitely Many Priorities},
  journal      = {Log. Methods Comput. Sci.},
  volume       = {2},
  number       = {4},
  year         = {2006},
  url          = {https://doi.org/10.2168/LMCS-2(4:6)2006},
  doi          = {10.2168/LMCS-2(4:6)2006},
  bibsource    = {dblp computer science bibliography, https://dblp.org}
}

@proceedings{DBLP:conf/dagstuhl/2001automata,
  editor       = {Erich Gr{\"{a}}del and
                  Wolfgang Thomas and
                  Thomas Wilke},
  title        = {Automata, Logics, and Infinite Games: {A} Guide to Current Research},
  series       = {Lecture Notes in Computer Science},
  volume       = {2500},
  publisher    = {Springer},
  year         = {2002},
  url          = {https://doi.org/10.1007/3-540-36387-4},
  doi          = {10.1007/3-540-36387-4},
  isbn         = {3-540-00388-6},
  bibsource    = {dblp computer science bibliography, https://dblp.org}
}

@inproceedings{DBLP:conf/cav/Cook18,
  author       = {Byron Cook},
  editor       = {Hana Chockler and
                  Georg Weissenbacher},
  title        = {Formal Reasoning About the Security of {A}mazon Web Services},
  booktitle    = {Computer Aided Verification - 30th International Conference, {CAV}
                  2018, Held as Part of the Federated Logic Conference, FloC 2018, Oxford,
                  UK, July 14-17, 2018, Proceedings, Part {I}},
  series       = {Lecture Notes in Computer Science},
  pages        = {38--47},
  publisher    = {Springer},
  year         = {2018},
  url          = {https://doi.org/10.1007/978-3-319-96145-3\_3},
  doi          = {10.1007/978-3-319-96145-3\_3},
  bibsource    = {dblp computer science bibliography, https://dblp.org}
}

@misc{CI26,
      title={An algebraic characterisation of {E}ve-positional languages}, 
      author={Thomas Colcombet and Olivier Idir},
      year={2026},
      eprint={2604.22648},
      archivePrefix={arXiv},
      primaryClass={cs.FL},
      url={https://arxiv.org/abs/2604.22648}, 
}

\end{document}